\pdfoutput=1
\documentclass[letterpaper, 10 pt, conference]{ieeeconf}  % Comment this line out if you need a4paper

\IEEEoverridecommandlockouts                              % This command is only needed if 
\usepackage{graphics} % for pdf, bitmapped graphics files
\usepackage{epsfig} % for postscript graphics files
\usepackage{algorithm}
\usepackage{algpseudocode}
\usepackage{optidef}
\usepackage{amsmath} % assumes amsmath package installed
\usepackage{amssymb}  % assumes amsmath package installed
\usepackage{csquotes}
\usepackage{url}
\newtheorem{theorem}{Theorem}
\newtheorem{crl}{Corollary}

\allowdisplaybreaks
\title{\LARGE \bf
Constrained Optimal Tracking Control of Unknown Systems: \\A Multi-Step Linear Programming Approach*
}

\author{Alexandros Tanzanakis and John Lygeros$^{1}$% <-this % stops a space
\thanks{*This research work was supported by the European Research Council (ERC) under the project OCAL, grant number 787845.}% <-this % stops a space
\thanks{$^{1}$All authors are with the Department of Information Technology and Electrical Engineering, ETH Zurich, Switzerland, {\tt\small \{atanzana,jlygeros\}@ethz.ch}}%
}

\begin{document}

\maketitle
\thispagestyle{empty}
\pagestyle{empty}

%%%%%%%%%%%%%%%%%%%%%%%%%%%%%%%%%%%%%%%%%%%%%%%%%%%%%%%%%%%%%%%%%%%%%%%%%%%%%%%%
\begin{abstract}
We study the problem of optimal state-feedback tracking control for unknown discrete-time deterministic systems with input constraints. To handle input constraints, state-of-art methods utilize a certain nonquadratic stage cost function, which is sometimes limiting real systems. Furthermore, it is well known that Policy Iteration (PI) and Value Iteration (VI), two widely used algorithms in data-driven control, offer complementary strengths and weaknesses. In this work, a two-step transformation is employed, which converts the constrained-input optimal tracking problem to an unconstrained augmented optimal regulation problem, and allows the consideration of general stage cost functions. Then, a novel multi-step VI algorithm based on Q-learning and linear programming is derived. The proposed algorithm improves the convergence speed of VI, avoids the requirement for an initial stabilizing control policy of PI, and computes a constrained optimal feedback controller without the knowledge of a system model and stage cost function. Simulation studies demonstrate the reliability and performance of the proposed approach.
\end{abstract}
%%%%%%%%%%%%%%%%%%%%%%%%%%%%%%%%%%%%%%%%%%%%%%%%%%%%%%%%%%%%%%%%%%%%%%%%%%%%%%%%
\section{INTRODUCTION}
Reinforcement Learning (RL) \cite{c1} is currently in the spotlight of the modern control engineering community. Closely related to approximate dynamic programming (ADP), it seeks to approximately solve the Hamilton-Jacobi-Bellman (HJB) or Bellman equations, and hence obtain an approximate optimal control policy. The core ingredients in many state-of-art approaches are Q-learning \cite{c2} and neurodynamic programming \cite{c3}, which provide reliable model-free optimal control. Due to the complex nature of solving the Bellman/HJB equations analytically \cite{c4}, \cite{c5}, ADP utilizes an Actor-Critic structure \cite{c6}, \cite{c7} for the approximation of the value function or Q-function and control policy with parametric function approximators (e.g. neural networks). The Bellman/HJB equations can be approximately solved by using iterative methods such as Policy Iteration (PI) and Value Iteration (VI) \cite{c8}-\cite{c10}, \cite{c30}, \cite{c35}. These two classes of methods offer complementary strengths and weaknesses. PI provides fast convergence to an optimal control policy, while it requires an initial stabilizing control policy, which is generally difficult to derive. On the other hand, VI avoids the requirement for an initial stabilizing control policy, at the cost of slower convergence.\\
In recent years, there is an increasing interest in deriving improved variants of PI and VI. A model-based PI algorithm utilizing a limited lookahead policy evaluation approach is presented in \cite{c11}. A class of PI algorithms based on temporal difference learning and the $\lambda$-operator is proposed in \cite{c12}, which has been further extended using abstract dynamic programming \cite{c13} and randomized proximal methods \cite{c14}, \cite{c15}. An alternative family of model-based tabular PI algorithms with multi-step greedy policy improvement is derived in \cite{c16}, \cite{c17}. In all these works, the requirement for an initial stabilizing control policy still holds. An improved model-based backward VI algorithm is presented in \cite{c18}. A data-driven VI formulation using finite horizon policy evaluation similar to the model-based limited lookahead approach in \cite{c11} is presented in \cite{c19}, showing good performance on undiscounted regulation of linear-time-invariant (LTI) systems without input constraints.\\ 
Input saturation is an inevitable property of many real systems. To handle input constraints, many research works \cite{c20}-\cite{c24} consider the use of a specific nonquadratic stage cost function, based on which the optimal control policy can be derived in closed-form. However, this approach can be applied to LTI and affine nonlinear systems, but not to nonaffine nonlinear systems. This is because the optimal control policy of a nonaffine nonlinear system is an implicit function of the value function or Q-function, which is generally very challenging to be computed analytically \cite{c30}. Furthermore, the use of the same stage cost function in all constrained optimal control problems is an unrealistic approach, since it cannot reflect the true cost incurred during the operation of a system.\\
The Linear Programming (LP) approach to ADP provides a promising model-based optimization method for the approximation of the optimal value function \cite{c25} or Q-function \cite{c26}, \cite{c27} and control policy, with provable theoretical guarantees on the quality of approximation and algorithmic performance. The authors in \cite{c28}, \cite{c29} have extended the LP approach to model-free optimal regulation problems, by deriving data-driven PI and VI based LP algorithms. However, to the best of our knowledge, the problem of optimal tracking control (with or without input constraints) has never been studied in the current model-based and model-free LP literature, remaining an important open problem.\\    
In this work, we address the constrained optimal tracking problem through a two-step transformation scheme. This converts the initial constrained-input optimal tracking problem to an unconstrained augmented optimal regulation problem, and allows the use of general stage cost functions. Then, we derive a novel multi-step VI algorithm for general discounted constrained-input optimal tracking control problems, based on Q-learning and linear programming. The proposed algorithm improves the convergence speed of VI, avoids the requirement for an initial stabilizing control policy of PI, and provides rigorous theoretical guarantees. Simulation studies on a nonlinear system confirm the effectiveness of the proposed approach.\\   
The structure of the paper is given as follows. The problem formulation is given in Section II. The proposed algorithm is derived and analyzed in Sections III and IV. Simulation studies are carried on Section V and conclusions are given in Section VI.\\
\textbf{Notation.} $\mathbb{N}$ defines the set of natural numbers, while $\mathbb{N}_{0}$ the set of natural numbers including $0$. $\mathbb{R_{+}}$ refers to the set of non-negative real numbers. $I_{n}$ defines an identity matrix of size $n\times n$. $\mathbb{S}_{n}$ defines the set of symmetric matrices of size $n\times n$. Finally, $tr(A)$ denotes the trace of a matrix $A\in \mathbb{R}^{n\times n}$.  
\section{PROBLEM FORMULATION}
In this work, we consider discrete-time deterministic systems in state-space form,
\begin{equation}
x_{k+1} = f(x_k,u_k),
\end{equation}
where $x_k\in \mathcal{X}\subseteq \mathbb{R}^{n}$, $u_k \in \mathcal{W}\subseteq \mathcal{U}\subseteq \mathbb{R}^{m}$ and $f:\mathcal{X}\times \mathcal{W}\rightarrow \mathcal{X}$ denote the state vector, input vector and system dynamics at time $k\in \mathbb{N}_{0}$ respectively. We assume that $u_k=\begin{bmatrix} u_{1,k} & u_{2,k}\ldots u_{m,k}\end{bmatrix}^{T}$ is constrained by $|u_{j,k}|\leq \bar{u}_{j}$, with $\bar{u}_{j}>0$ for all $j$. Finally, $r_k\in \mathcal{X}$ is a desired bounded reference signal generated by an \enquote{exosystem}
\begin{equation}
r_{k+1} = g(r_k),
\end{equation}
where $g:\mathcal{X}\rightarrow \mathcal{X}$ is the reference generation function.
\subsection{A two-step transformation scheme}
We now derive a two-step transformation that converts the initial problem to an unconstrained augmented optimal regulation problem. The first step consists of defining the tracking error and augmented state variables as $e_k = x_k-r_k$ and $z_k = \begin{bmatrix} e^{T}_{k} & r^{T}_{k}\end{bmatrix}^{T}\in \mathcal{Z}\subseteq \mathbb{R}^{2n}$ respectively. Thus, we derive an augmented system given by
\begin{equation}
z_{k+1} = F(z_{k},u_{k}) = \begin{bmatrix} f(e_{k}+r_{k},u_{k})-g(r_{k}) \\ g(r_{k}) \end{bmatrix}.
\end{equation}
For the second step, following \cite{c30}, we assume that the constrained-input state-feedback control policy $u$ can be defined as
\begin{equation}
u = s\big(\mu(z)\big),
\end{equation} 
where $\mu(z):\mathcal{Z}\rightarrow \mathcal{U}$ is an unconstrained state-feedback control policy and $s:\mathcal{U}\rightarrow \mathcal{W}$ is some form of saturation function \cite{c20}-\cite{c24}. The main objective then becomes the computation of an unconstrained state-feedback control policy $\mu(z)$ that minimizes a general cost function,
\begin{equation}
\mathcal{J}^{\mu}(z_0)=\sum_{k=0}^{\infty}\gamma^{k}L\big(z_k,\mu(z_k)\big),
\end{equation}
where $\gamma \in (0,1]$ is the discount factor, $L\big(z,\mu(z)\big) \overset{\Delta}{=} l\big(z,s\big(\mu(z)\big)\big),$ and $l:\mathcal{Z}\times \mathcal{W}\rightarrow \mathbb{R}_{+}$ is the stage cost function. We note that, as usual in optimal tracking control problems, if $r_k$ is a general bounded signal, a discount factor $\gamma \in (0,1)$ is required for the cost function (5) to be finite \cite{c31}-\cite{c35}. Furthermore, a discount factor $\gamma=1$ is only permitted if $g(\cdot)$ is apriori known to be an asympotically stable reference generation function.\\
To ensure the well posedness of the problem, we assume that there exists a policy $\mu(z)$ such that $\mathcal{J}^{\mu}(z)<\infty$ for all $z\in \mathcal{Z}$. Hereafter, we refer to policies with this property as stabilizing policies. We consider the case where the mathematical expressions of $f$, $g$ and $l$ (and therefore $F$ and $L$) are unknown, although their values can be observed by applying sequences of states $z$ and control policies $u$ through interaction with the system, e.g. simulations and experiments.
\subsection{Q-learning}
A fundamental element in Q-learning is the Q-function,
\begin{align}
Q^{\mu}(z,a) &= L(z,a)+\gamma \mathcal{J}^{\mu}(z_1),
\end{align}
where $a\in \mathcal{U}$ and $z_1=F\big(z,s(a)\big)$. To compute the optimal policy, we firstly need to derive the optimal Q-function, 
\begin{align}
Q^{\star}(z,a) = L(z,a) + \gamma \inf_{\mu}\mathcal{J}^{\mu}(z_1),
\end{align}
that can be shown to be the solution of the Bellman optimality equation \cite{c10},
\begin{equation}
Q^{*}(z,a)=\underbrace{L(z,a)+\gamma \min_{v} Q^{*}\big(z_1,v\big)}_{\mathcal{D}Q^{*}(z,a)}.
\end{equation} 
Here, $\mathcal{D}$ is the Q-function based Bellman operator that is both contractive and monotone \cite{c26}, \cite{c27}, ensuring the existence and uniqueness of the fixed point $Q^{\star}$ in (8). The unconstrained optimal policy can then be derived as
\begin{equation}
\mu^{*}(z) = \underset{v}{\mathrm{argmin }}Q^{*}(z,v).
\end{equation}
Once $\mu^{\star}$ has been found, the constrained optimal state-feedback controller can be easily computed based on (4).
\section{MULTI-STEP VI FOR CONSTRAINED-INPUT OPTIMAL TRACKING CONTROL}
In this section, we derive a Q-learning based multi-step VI algorithm for effective constrained-input optimal tracking control of unknown systems. Starting from an initial Q-function $Q^{0}(z,a)$ and iterating through index $i\geq 0$, it proceeds with the following till convergence of the Q-function:
\begin{enumerate}
\item[i)] \textbf{Policy Improvement:}
\begin{align} \mu^{i}(z)=\underset{v}{\mathrm{argmin}}Q^{i}(z,v). \end{align}
\item[ii)] \textbf{Policy Evaluation:} Solve for $Q^{i+1}$,
\begin{align}
Q^{i+1}(z,a)=& L(z,a)+\sum_{l=1}^{H_i-1}\gamma^{l}L\big(z_l,\mu^{i}(z_l)\big)\nonumber \\
&+\gamma^{H_i} Q^{i}\big(z_{H_i},\mu^{i}(z_{H_i})\big),
\end{align}
where $z_{1}$$=F\big(z,s(a)\big)$ and $z_{l}=F\big(z_{l-1},s\big(\mu^{i}(z_{l-1})\big)\big)$ for $l>1$.
\end{enumerate}
Compared to conventional PI and VI algorithms, the proposed policy evaluation scheme is different. At each iteration $i$, the policy evaluation of PI and VI uses single-step information based on the current and next time instants \cite{c8}-\cite{c10}, \cite{c30}, \cite{c35}. However, the policy evaluation (11) utilizes finite horizon (i.e. \enquote{multi-step}) data characterized by the horizon length $H_i\geq 1$; by setting $H_i=1$ for all $i$, the proposed algorithm is converted to the standard VI algorithm. Therefore, starting from the conventional VI algorithm and based on the idea of exploiting more data in the policy evaluation scheme at each iteration, the proposed multi-step VI algorithm aims to boost the convergence speed of VI, without requiring an initial stabilizing policy (as in PI).\\
The following theorem provides important monotonicity and convergence guarantees of the proposed algorithm.\\ 
\begin{theorem}
Consider the sequence $\{Q^{i}(z,a)\}$ generated by (10) and (11). If
\begin{equation}
Q^{0}(z,a)\geq L(z,a)+\gamma \min_{v}Q^{0}(z_1,v),
\end{equation}
where $z_1=F\big(z,s(a)\big)$, then
\begin{equation}
\begin{aligned}
Q^{i+1}(z,a)\leq L(z,a)+\gamma \min_{v}Q^{i}(z_1,v) \leq Q^{i}(z,a), \text{ for all $i$}.
\end{aligned}
\end{equation}
Furthemore,
\begin{equation}
\lim_{i\rightarrow \infty}Q^{i}(z,a)=Q^{\star}(z,a).
\end{equation}
\end{theorem}
\begin{proof} 
We firstly apply mathematical induction to prove (13). Based on (11) and (12), we get
\begin{align}
Q^{1}(z,a)=& L(z,a)+\sum_{l=1}^{H_0-1}\gamma^{l}L\big(z_l,\mu^{0}(z_l)\big) \nonumber \\
&+\gamma^{H_0} Q^{0}\big(z_{H_0},\mu^{0}(z_{H_0})\big) \nonumber\\
=&L(z,a)+\sum_{l=1}^{H_0-2}\gamma^{l}L\big(z_l,\mu^{0}(z_l)\big) \nonumber\\
&+\gamma^{H_0-1} L\big(z_{H_0-1},\mu^{0}(z_{H_0-1})\big)\nonumber\\
&+\gamma^{H_0} Q^{0}\big(z_{H_0},\mu^{0}(z_{H_0})\big)\nonumber\\
\overset{(10)}{=} &L(z,a)+\sum_{l=1}^{H_0-2}\gamma^{l}L\big(z_l,\mu^{0}(z_l)\big)\nonumber\\
&+\gamma^{H_0-1}\bigg(L\big(z_{H_0-1},\mu^{0}(z_{H_0-1})\big)\nonumber\\
&+ \gamma \min_{v}Q^{0}(z_{H_0},v)\bigg)\nonumber\\
\overset{(12)}{\leq}& L(z,a)+\sum_{l=1}^{H_0-2}\gamma^{l}L\big(z_l,\mu^{0}(z_l)\big)\nonumber\\
&+\gamma^{H_{0}-1}Q^{0}\big(z_{H_{0}-1},\mu^{0}(z_{H_{0}-1})\big)\nonumber.
\end{align}
Iterating leads to
\begin{align}
Q^{1}(z,a) \leq & L(z,a)+\gamma Q^{0}\big(z_{1},\mu^{0}(z_{1})\big) \nonumber\\
\overset{(10)}{=} &L(z,a)+\gamma \min_{v}Q^{0}\big(z_{1},v)\big).
\end{align}
Hence, by (12),
\begin{align}
Q^{1}(z,a)\leq L(z,a)+\gamma \min_{v}Q^{0}(z_{1},v) \leq Q^{0}(z,a).\nonumber
\end{align}
Then, we assume that (13) holds for $i-1$,
\begin{equation}
\begin{aligned}
Q^{i}(z,a)\leq L(z,a)+\gamma \min_{v}Q^{i-1}(z_1,v) \leq Q^{i-1}(z,a).
\end{aligned}
\end{equation}
It follows that
\begin{align}
Q^{i}(z,a)=& L(z,a)+\sum_{l=1}^{H_{i-1}-1}\gamma^{l}L\big(z_l,\mu^{i-1}(z_l)\big) \nonumber \\
&+\gamma^{H_{i-1}} Q^{i-1}\big(z_{H_{i-1}},\mu^{i-1}(z_{H_{i-1}})\big) \nonumber \\
\overset{(16)}{\geq} & L(z,a)+\sum_{l=1}^{H_{i-1}-1}\gamma^{l}L\big(z_l,\mu^{i-1}(z_l)\big) \nonumber \\
&+\gamma^{H_{i-1}}\bigg(L\big(z_{H_{i-1}},\mu^{i-1}(z_{H_{i-1}})\big)\nonumber\\
&+\gamma \min_{v}Q^{i-1}(z_{H_{i-1}+1},v)\bigg)\nonumber \\
\overset{(10)}{=}&L(z,a)+\sum_{l=1}^{H_{i-1}}\gamma^{l}L\big(z_l,\mu^{i-1}(z_l)\big) \nonumber \\
&+\gamma^{H_{i-1}+1} Q^{i-1}\big(z_{H_{i-1}+1},\mu^{i-1}(z_{H_{i-1}+1})\big)\nonumber \\
=&L(z,a)+\gamma\bigg(\sum_{l=1}^{H_{i-1}}\gamma^{l-1}L\big(z_l,\mu^{i-1}(z_l)\big) \nonumber \\
&+\gamma^{H_{i-1}} Q^{i-1}\big(z_{H_{i-1}+1},\mu^{i-1}(z_{H_{i-1}+1})\big)\bigg)\nonumber \\
\overset{(11)}{=} &L(z,a)+\gamma Q^{i}\big(z_{1},\mu^{i-1}(z_{1})\big)\nonumber \\
\geq & L(z,a)+\gamma \min_{v}Q^{i}(z_{1},v).
\end{align}
The reasoning leading up to (15) then gives,
\begin{align}
Q^{i+1}(z,a)=& L(z,a)+\sum_{l=1}^{H_i-1}\gamma^{l}L\big(z_l,\mu^{i}(z_l)\big) \nonumber \\
&+\gamma^{H_i} Q^{i}\big(z_{H_i},\mu^{i}(z_{H_i})\big) \nonumber\\
\leq & L(z,a)+\sum_{l=1}^{H_i-2}\gamma^{l}L\big(z_l,\mu^{i}(z_l)\big)\nonumber\\
&+\gamma^{H_{i}-1}Q^{i}\big(z_{H_{i}-1},\mu^{i}(z_{H_{i}-1})\big)\nonumber\\
\leq & L(z,a)+\gamma Q^{i}\big(z_{1},\mu^{i}(z_{1})\big) \nonumber\\
\overset{(10)}{=}& L(z,a)+\gamma \min_{v}Q^{i}(z_{1},v).
\end{align}
Based on (17) and (18), we finally conclude that (13) holds for all $i$.
\\For the second part of the proof, based on (13), the sequence $\{Q^{i}(z,a)\}$ is non-increasing and, since $L$ is assumed to be non-negative, lower bounded by $0$ for all $i$; hence it has a point-wise limit $Q^{\infty}(z,a) = \lim_{i\rightarrow \infty} Q^{i}(z,a)$ for all $(z,a)\in \mathcal{Z}\times \mathcal{U}$. Furthermore, let $\mu^{\infty}(z) = \underset{v}{\mathrm{argmin }}Q^{\infty}(z,v)$. Then, taking the limit of (13), we have that
\begin{align}
Q^{\infty}(z,a) \leq L(z,a)+\gamma \min_{v}Q^{\infty}(z_1,v) \leq Q^{\infty}(z,a),\nonumber
\end{align}
i.e. 
\begin{align}
Q^{\infty}(z,a)&= L(z,a)+\gamma \min_{v}Q^{\infty}(z_1,v).
\end{align}
We note that (19) is the Bellman optimality equation (8), which is satisfied by $Q^{\infty}$. Therefore, $Q^{\infty}(z,a)=Q^{\star}(z,a)$.
\end{proof}
\begin{crl}
Let $\{Q^{i}(z,a)\}$ generated by (10) and (11), and $\{Q^{i}_{VI}(z,a)\}$ the sequence generated by the standard VI algorithm \big(corresponding to setting $H_i=1$ for all $i$ in (11)\big). Under assumption (12), if $Q^{0}(z,a)\leq Q^{0}_{VI}(z,a)$, then $Q^{i}(z,a)\leq Q^{i}_{VI}(z,a)$ for all $i$.
\end{crl}
\begin{proof}
We apply mathematical induction to prove the desired statement. It is clear that the statement holds for $i=0$. We then assume that $Q^{i}(z,a)\leq Q^{i}_{VI}(z,a)$ holds for $i$. Then, based on (13), we have that
\begin{align}
Q^{i+1}(z,a)&\leq L(z,a)+\gamma \min_{v}Q^{i}(z_1,v) \nonumber\\
&\leq L(z,a)+\gamma \min_{v}Q^{i}_{VI}(z_1,v) \nonumber\\
&= Q^{i+1}_{VI}(z,a).
\end{align}
\end{proof}
In other words, the proposed multi-step VI algorithm provides a better Q-function estimate than conventional VI at each iteration $i$.\\
\begin{crl}
Consider a Q-function $\bar{Q}(z,a)$ which satisfies condition (13), and its associated control policy $\mu_{\bar{Q}}(z)$ which is computed based on (10). Let $h_1,h_2\geq 1$ be two distinct horizon lengths, which formulate the following realizations of the policy evaluation step (11),
\begin{align}
Q_{h_1}(z,a) =& L(z,a)+\sum_{l=1}^{h_1-1}\gamma^{l}L\big(z_l,\mu_{\bar{Q}}(z_l)\big) \nonumber \\
&+\gamma^{h_1} \bar{Q}\big(z_{h_1},\mu_{\bar{Q}}(z_{h_1})\big), \nonumber\\
Q_{h_2}(z,a) =& L(z,a)+\sum_{l=1}^{h_2-1}\gamma^{l}L\big(z_l,\mu_{\bar{Q}}(z_l)\big) \nonumber \\
&+\gamma^{h_2} \bar{Q}\big(z_{h_2},\mu_{\bar{Q}}(z_{h_2})\big)\nonumber.
\end{align}
If $h_1\geq h_2$, then $Q_{h_1}(z,a)\leq Q_{h_2}(z,a)$.
\end{crl}
\begin{proof}
Similar to the reasoning for the proof of (15), we get that
\begin{align}
Q_{h_1}(z,a)=& L(z,a)+\sum_{l=1}^{h_1-1}\gamma^{l}L\big(z_l,\mu_{\bar{Q}}(z_l)\big) \nonumber \\
&+\gamma^{h_1} \bar{Q}\big(z_{h_1},\mu_{\bar{Q}}(z_{h_1})\big) \nonumber\\
\overset{(10)}{=} & L(z,a)+\sum_{l=1}^{h_1-2}\gamma^{l}L\big(z_l,\mu_{\bar{Q}}(z_l)\big)\nonumber\\
&+\gamma^{h_1-1}\bigg(L\big(z_{h_1-1},\mu_{\bar{Q}}(z_{h_1-1})\big)\nonumber\\
&+ \gamma \min_{v}\bar{Q}(z_{h_1},v)\bigg)\nonumber\\
\overset{(13)}{\leq} & L(z,a)+\sum_{l=1}^{h_1-2}\gamma^{l}L\big(z_l,\mu_{\bar{Q}}(z_l)\big)\nonumber\\
&+\gamma^{h_1-1}\bar{Q}\big(z_{h_1-1},\mu_{\bar{Q}}(z_{h_1-1})\big).\nonumber
\end{align}
Iterating leads to
\begin{align}
Q_{h_1}(z,a) \leq & L(z,a)+\sum_{l=1}^{h_2-1}\gamma^{l}L\big(z_l,\mu_{\bar{Q}}(z_l)\big)\nonumber\\
&+\gamma^{h_2}\bar{Q}\big(z_{h_2},\mu_{\bar{Q}}(z_{h_2})\big)\nonumber\\
= &Q_{h_2}(z,a).
\end{align}  
\end{proof}
\textbf{Remark 1.} We note that the initialization condition for $Q^{0}(z,a)$ (12) is a sufficient condition. Since $H_i\geq 1$ for all $i$, we can set $H_0=1$, transforming the proposed multi-step algorithm to the standard VI algorithm. Due to the fact that the sequence $\{Q^{i}(z,a)\}$ is non-increasing based on Theorem 1, we can take advantage of the monotonicity properties of VI \cite{c8}-\cite{c10}, \cite{c30} and simply initialize $Q^{0}(z,a)$ with an arbitrary, sufficiently large positive definite function. Then, based on Corollary $2$, we can increase $H_i$ for $i>0$ to boost the convergence speed. Therefore, the initialization condition (12) can be relaxed, while the monotonicity and convergence guarantees of Theorem 1 can still hold. The derivation of such theoretical relaxations will be investigated in future work. 
\section{A DATA-DRIVEN, MULTI-STEP \\LP FORMULATION}
In this section, we proceed with the reformulation of the policy evaluation scheme (11) as a data-driven optimization problem. We start by relaxing the Bellman equation (8) to an inequality \cite{c26}-\cite{c29},
\begin{align}
Q(z,a)&\leq \mathcal{D}Q(z,a),\enspace \forall (z,a)\in \mathcal{Z}\times \mathcal{U} \nonumber \\
&\leq L(z,a)+\gamma Q(z_1,v) \nonumber \\
&\forall (z,a,v)\in \mathcal{Z}\times \mathcal{U}^{2}
\end{align}
where $z_1=F\big(z,s(a)\big)$, and reformulate (8) as an equivalent linear program,
\begin{equation}
\begin{aligned}
& \underset{Q\in \mathcal{F}(\mathcal{Z},\mathcal{U})}{\text{max}} 
& & \int_{\mathcal{Z}\times \mathcal{U}}Q(z,a)c(dz,da)\\
& \text{s.t.}
& &Q(z,a) \leq L(z,a)+ \gamma Q\big(z_1,v\big)\\
& & &\forall (z,a,v)\in \mathcal{Z}\times \mathcal{U}^{2},\\
\end{aligned}
\end{equation} 
where $\mathcal{F}(\mathcal{Z},\mathcal{U})$ is the space of measurable, bounded (in a suitable weighted norm) Q-functions. The state-control policy relevance weight $c(\cdot,\cdot)$ is a probability measure that allocates positive mass to all open subsets of $\mathcal{Z}\times \mathcal{U}$ \cite{c27}.\\

\textit{Lemma 1 \cite{c26}-\cite{c29}:} If $Q^{\star}\in \mathcal{F}(\mathcal{Z},\mathcal{U})$, then the optimizer of (23) is the same as the solution of (8), for $c$ almost all $(z,a)\in \mathcal{Z}\times \mathcal{U}$.\\
In general, computing the optimizer of (23) is intractable, due to the curse of dimensionality. The main challenges are described in \cite{c26}-\cite{c29}:
\begin{enumerate}
\item[$i)$] $\mathcal{F}(\mathcal{Z},\mathcal{U})$ is a high-dimensional (if $\mathcal{Z}, \mathcal{U}$ are finite) or an infinite dimensional (if $\mathcal{Z}, \mathcal{U}$ are not finite) space.
\item[$ii)$] The optimization problem (23) involves a large or infinite number of inequality constraints.
\item[$iii)$] Due to the fact that $Q^{\star}\in \mathcal{F}(\mathcal{Z},\mathcal{U})$, the policy improvement (9) is generally intractable.
\end{enumerate}
To deal with $i)$, we consider a restricted function space $\hat{\mathcal{F}}(\mathcal{Z}\times \mathcal{U})$ spanned by a finite number of basis functions $\hat{Q}_j(z,a)$ for $j=1,\ldots,K$,
\begin{equation*}
\hat{\mathcal{F}}(\mathcal{Z}\times \mathcal{U})=\{Q(\cdot,\cdot)|Q(z,a) =\alpha^{T}\hat{Q}(z,a)\},
\end{equation*}
where $\alpha \in \mathbb{R}^{K}$ and $\hat{Q}(z,a):\mathcal{Z}\times \mathcal{U}\rightarrow \mathbb{R}^{K}$. An approximate solution to (8) can then be derived by solving the following linear program,
\begin{equation}
\begin{aligned}
& \underset{Q\in \hat{\mathcal{F}}(\mathcal{Z}\times \mathcal{U})}{\text{max}} 
& & \int_{\mathcal{Z}\times \mathcal{U}}Q(z,a)c(dz,da)\\
& \text{s.t.}
& &Q(z,a) \leq L(z,a)+ \gamma Q\big(z_1,v\big)\\
& & &\forall (z,a,v)\in \mathcal{Z}\times \mathcal{U}^{2}.\\
\end{aligned}
\end{equation}
The main challenge of course is the choice of an appropriate function space $\hat{\mathcal{F}}(\mathcal{Z}\times \mathcal{U})$. In order to deal with $ii)$, if the functional forms of $f$, $g$ and $l$ (and therefore $F$ and $L$) are unknown, we can proceed with the relaxation of the inequality constraints in (24) based on sampling methods \cite{c28}, \cite{c29}. In particular, we can construct a buffer of data samples $\{z_b,a_b,L(z_b,a_b),z_{1,b},v_b\}_{b=1}^{B}$ through simulations and experiments with the system, where $B\in \mathbb{N}$ is the size of the buffer. Therefore, based on the constructed buffer, we can replace the inequality constraints in (24) by their sampled variants. Finally, in order to deal with $iii)$, we can restrict $\hat{\mathcal{F}}(\mathcal{Z}\times \mathcal{U})$ to a family of basis functions which are convex in $a$.\\
\textbf{Remark 2.} The quality of solution of (24) generally depends on the specific realization of $c(\cdot,\cdot)$ \cite{c26}-\cite{c29}. However, based on Lemma 1, if $Q^{\star}\in \hat{\mathcal{F}}(\mathcal{Z}\times \mathcal{U})$, then the solution of (24) is $Q^{\star}$, provided that $c(\cdot,\cdot)$ allocates positive mass to all open subsets of $\mathcal{Z}\times \mathcal{U}$.\\
Based on the LP reformulation of the Bellman equation (8) discussed above, we can similarly reformulate the policy evaluation (11) as a data-driven linear program which is given by (25). Furthermore, based on Remark 2, the resulting linear program inherits all the properties of the multi-step algorithm which is developed and analyzed in Section III. Algorithm 1 shows the final algorithm, which we call MSQ-VI-LP. \\
\textbf{Remark 3.} An interesting feature of the proposed algorithm is that, while the use of a horizon length $H_i$ implies more data is available, the number of decision variables and inequality constraints in (25) does not depend on $H_i$, but just on the buffer size $B$ and richness of the function class $\hat{\mathcal{F}}(\mathcal{Z}\times \mathcal{U})$. All data in the buffer does get used, without increasing the complexity of the linear program.  
\begin{algorithm}
        \caption{MSQ-VI-LP algorithm.}\label{euclid}
        \begin{algorithmic}[1]
            \State Select $\epsilon >0$ and $B \in \mathbb{N}$.
            \State Construct $\{z_b,a_b,L(z_b,a_b)\}_{b=1}^{B}$.
            \State Choose $Q^0(z,a)$, based on Remark 1.
             \State Set $i=0$.
              \State Select $H_i\geq 1$.
               \State $\mu^{i}(z)=\underset{v}{\mathrm{argmin}}Q^{i}(z,v).$
              \State Collect $\big\{[z_{l,b}]_{l=1}^{H_i},[L(z_{l,b},\mu^{i}(z_{l,b})]_{l=1}^{H_i-1},\mu^{i}(z_{H_i,b})\big\}_{b=1}^{B}$ through interactions with the system, where $z_{1,b}=F\big(z_b,s(a_b)\big)$ and $z_{l,b}=F\big(z_{l-1,b},s\big(\mu^{i}(z_{l-1,b})\big)\big)$ for $l>1$. 
             \State Solve the LP problem,
\begin{align}
& \underset{Q^{i+1}\in \hat{\mathcal{F}}(\mathcal{Z}\times \mathcal{U})}{\text{max}} 
& & \int_{\mathcal{Z}\times \mathcal{U}}Q^{i+1}(z,a)c(dz,da)\nonumber\\
& \text{s.t.}
& &Q^{i+1}(z_b,a_b) \leq L(z_b,a_b)\nonumber\\
& & & +\sum_{l=1}^{H_i-1}\gamma^{l}L\big(z_{l,b},\mu^{i}(z_{l,b})\big)\nonumber \\
& & &+\gamma^{H_i} Q^{i}\big(z_{H_i,b},\mu^{i}(z_{H_i,b})\big)\nonumber \\
& & & \text{for } b=1,\ldots,B.
\end{align}
\State If $\underset{b}{\mathrm{max}} |Q^{i+1}(z_b,a_b)-Q^{i}(z_b,a_b)|> \epsilon$, set $i=i+1$ and go to Step $5$.
\State Substitute $\mu^{i}$ into (4) and return the approximate constrained optimal controller.
        \end{algorithmic}
    \end{algorithm}
\\In model-free optimal adaptive control problems, the satisfaction of the persistance of excitation (PoE) condition \cite{c36} is required to achieve optimal parameter convergence. Based on \cite{c19}, \cite{c29}, \cite{c36}, an effective way to satisfy PoE in the proposed algorithm is to make a sufficiently rich set of state-control policy pairs $\{z_{b},a_{b}\}_{b=1}^{B}$. These pairs will initialize the policy evaluation of MSQ-VI-LP at each iteration $i$ and provide in total $B$ inequality constraints based on which the LP problem (25) can be solved. The richness of the pairs can be achieved in many ways, e.g. by randomly sampling $(z,a)$ based on suitable probability distributions. The choice of buffer size $B$ is application dependent, and a challenging system will generally require a large buffer. 
\section{SIMULATION STUDIES}
\subsection{Main Configuration}
Consider the two-dimensional nonlinear system from \cite{c37} with a minor modification,
\begin{equation*}
x_{k+1} = \begin{bmatrix} (x_{1,k}+x_{2,k}^{2}+u_k)\cos(x_{2,k})\\(2x_{1,k}^{2}+2x_{2,k}+2u_k)\sin(x_{2,k})\end{bmatrix},
\end{equation*}
where $\mathcal{X}= \mathbb{R}^{2}$, $\mathcal{U}=\mathbb{R}$, and $x_0=\begin{bmatrix} 0.8 \\ -1.1 \end{bmatrix}$. Let the desired reference signal $r_k$ be generated by the sine wave reference generator,
\begin{equation*}
r_{k+1} = \begin{bmatrix} 0.9751r_{1,k}+0.0992r_{2,k} \\ -0.4958r_{1,k}+0.9751r_{2,k} \end{bmatrix},
\end{equation*}
where $r_0 = \begin{bmatrix} 0.5 \\ 0.5 \end{bmatrix}$, making $\mathcal{Z}=\mathbb{R}^{4}$. A discount factor $\gamma=0.95$ is used, and the threshold parameter is set to $\epsilon=10^{-15}$. A quadratic stage cost function,
\begin{equation*}
l(z,u) = e^{T}Ee+u^{T}Fu,
\end{equation*}
is considered, where $E=4\cdot I_{2}$ and $F=1$. We conduct simulation studies on both the unconstrained and constrained-input optimal tracking control of the system. For the unconstrained case, we consider $u$ an unconstrained variable. For the constrained-input case, we impose constraints of the form $|u|\leq 0.7$, and use the hard saturation function,
\begin{equation*}
u = s\big(\mu(z)\big) = \begin{cases}
\mu(z),\text{ if }|\mu(z)|\leq 0.7,\\
0.7,\text{ if }\mu(z) >0.7,\\
-0.7,\text{ otherwise.}
\end{cases} 
\end{equation*}
The following family of Q-functions is employed,
\begin{equation*}
\hat{\mathcal{F}}(\mathcal{Z}\times \mathcal{U}) =\big\{Q(\cdot,\cdot)|Q(z,a)= \begin{bmatrix} e_1 \\e_2\\ r_1 \\r_2\\r_1^{2}\\r_2^{2}\\ a \end{bmatrix}^{T}\hat{P}\begin{bmatrix} e_1 \\e_2\\ r_1 \\r_2\\r_1^{2}\\r_2^{2}\\ a \end{bmatrix}  
\end{equation*} 
where $\hat{P}\in \mathbb{R}^{7\times 7}$. The state-control policy relevance weight $c(\cdot,\cdot)$ is considered a probability measure. The terms $r_{1}^{2}, r_2^{2}$ in the Q-function representation require the first four moments of $c(\cdot,\cdot)$ to appear in the objective function of the linear program (25). Considering the first moment as $\mu_c=0_{5\times 1}$, the objective function of the LP problem (25) reduces to \cite{c26}-\cite{c27}
\begin{equation*}
\int_{\mathcal{Z}\times \mathcal{U}}Q(z,a)c(dz,da) = tr(\hat{P_1}\Sigma_c) +\hat{p_2}^{T}s_c+\hat{p_3}^{T}k_c,
\end{equation*}
where $\hat{P}_1\in \mathbb{R}^{5\times 5}$, $\hat{p}_2\in \mathbb{R}^{20}$ and $\hat{p}_3\in \mathbb{R}^{4}$ are elements of the matrix $\hat{P}$ with second, third and fourth moments given by $\Sigma_c \in \mathbb{S}_5$, $s_c \in \mathbb{R}^{20}$ and $k_c \in \mathbb{R}^{4}$ respectively. Here, we choose $\Sigma_c = I_{5}$, $s_c = 1_{20\times 1}$ and $k_c=1_{4\times 1}$.
\subsection{Comparison with state-of-art data-driven LP algorithms}
We compare the performance of the proposed MSQ-VI-LP algorithm with the state-of-art LP-based PI and VI algorithms in \cite{c29}, which we refer to as Q-PI-LP and Q-VI-LP respectively. These algorithms were originally developed for unconstrained optimal regulation of unknown systems, but can easily be extended to the optimal tracking control setting as in Section II.\\
The simulation study involves the comparison of the following algorithmic variants:
\begin{itemize}
\item The proposed MSQ-VI-LP algorithm, with both an initial arbitrary (i.e. MSQ-VI-LP[A]) and stabilizing control policy (i.e. MSQ-VI-LP[S]),
\item The Q-PI-LP algorithm \cite{c29}, which requires an initial stabilizing control policy, and
\item The Q-VI-LP algorithm \cite{c29}, with both an initial arbitrary (i.e. Q-VI-LP[A]) and stabilizing control policy (i.e. Q-VI-LP[S]).  
\end{itemize}
To satisfy the PoE condition in Q-PI-LP and Q-VI-LP, we utilize the off-policy learning method in \cite{c29} called Randomized Experience Replay (RER). RER proceeds with the construction of a buffer of data samples $\{z_n,a_n,y_n,L_n\}_{n=1}^{N}$, which is repeatedly used for the duration of the algorithms. Here, $N$ is the size of the RER buffer, $(z_n,a_n)$ are arbitrary state-control policy pairs, $y_n=F\big(z_n,s(a_n)\big)$ and $L_n=L(z_n,a_n)$. We set $N=2000$ and construct $(z_n,a_n)_{n=1}^{N}$ with $z_n \sim Uni(-5,5)$ and $a_n \sim Uni(-2,2)$. For the PoE condition in MSQ-VI-LP, we use the same set of pairs $(z,a)$ constructed above, and therefore $B=N=2000$. \\
To initialize Q-PI-LP, we use the stabilizing control policy $\mu^{0}(z) = \begin{bmatrix} -1.5 & 0.5 & 0 & 0&0&0\end{bmatrix}\begin{bmatrix} e \\ r \\ r^{2} \end{bmatrix}$ based on \cite{c37}, which ensures that $\mathcal{J}^{\mu^{0}}(z)<\infty$, for all $z\in \mathcal{Z}$. For MSQ-VI-LP[A] and Q-VI-LP[A], we select an arbitrary, sufficiently large positive definite initial matrix $\hat{P}^{0}$ given by
\begin{equation*}
\begingroup % keep the change local
\setlength\arraycolsep{0.8pt}
\begin{bmatrix} 34.49 & -1.88 & -0.36 & -9.25 & -6.86 & 11.84 & 3.97\\ -1.88 & 96.46 & 7.25 & 29.08 & -7.05 & -22.61 & -3.71 \\ -0.36 & 7.25 & 21.69 & 5.4 & -18.23 & 1.13 & 4.85 \\ -9.25 & 29.08 & 5.4 & 19.68 & -2.49 & -11.5 & -4.89\\ -6.86 & -7.05 & -18.23 & -2.49 & 39.83 & 1.64 & -13.31 \\ 11.84 & -22.61 & 1.13 & -11.5 & 1.64 & 22.86 & 3.38 \\ 3.97 & -3.71 & 4.85 & -4.89 & -13.31 & 3.38& 0.69 \end{bmatrix}
\endgroup
\end{equation*}    
which leads to a non-stabilizing initial control policy $\mu^{0}$. For MSQ-VI-LP[S] and Q-VI-LP[S], we use the same $\hat{P}^{0}$ as above, although we apply the initial stabilizing control policy of Q-PI-LP. Finally, for the adaptation of $H_i$ in MSQ-VI-LP, we use the tuning control law \cite{c19}
\begin{align}
H_i = 1+\bigg[K\sqrt{i}\bigg],
\end{align}
where $K\geq 0$ and $[\cdot]$ is the rounding of a decimal to the closest integer. The intuition for choosing such a tuning law is discussed in Remark $1$. In our case, we choose $K=5$.  
\begin{figure}[htp]
\begin{centering}$
\begin{array}{cc}
\hbox{\hspace{-0.4em}}\includegraphics[width=0.478\textwidth]{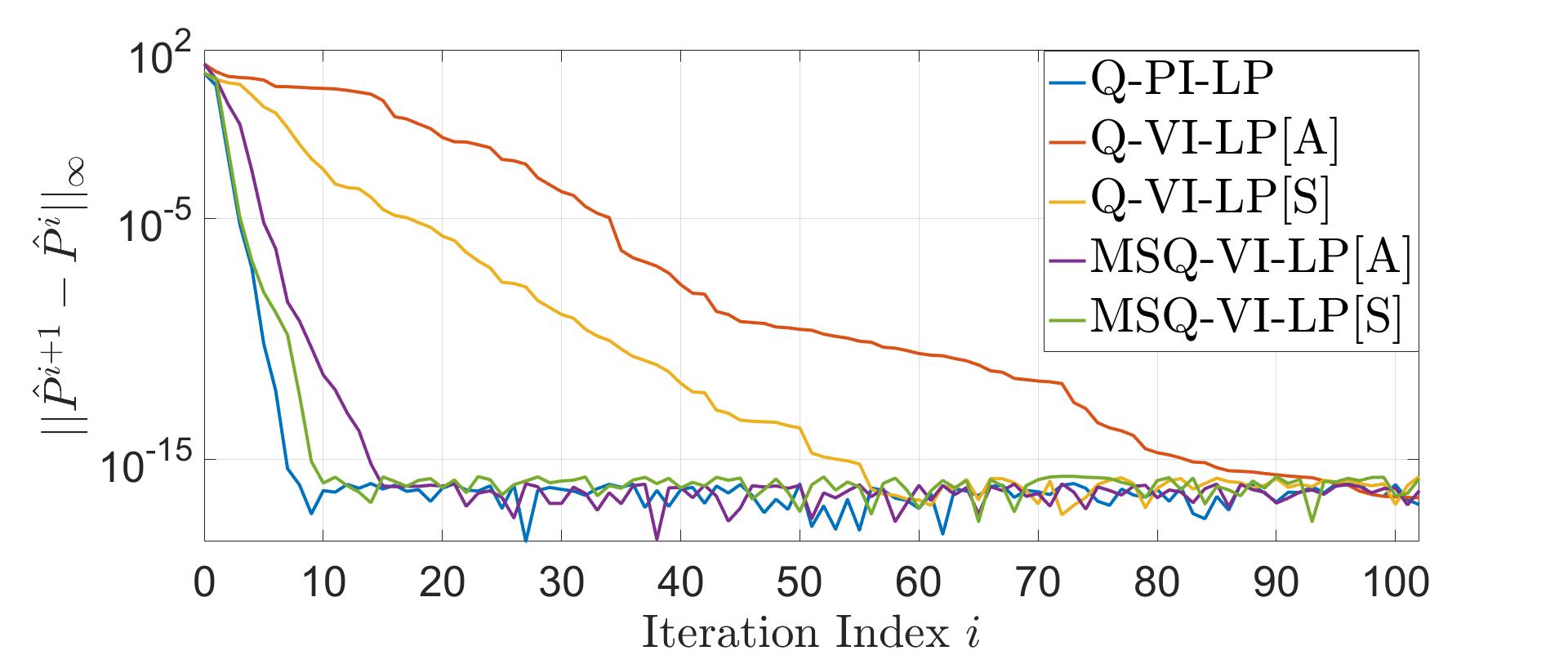} \\
\hbox{\hspace{-0.4em}}\includegraphics[width=0.478\textwidth]{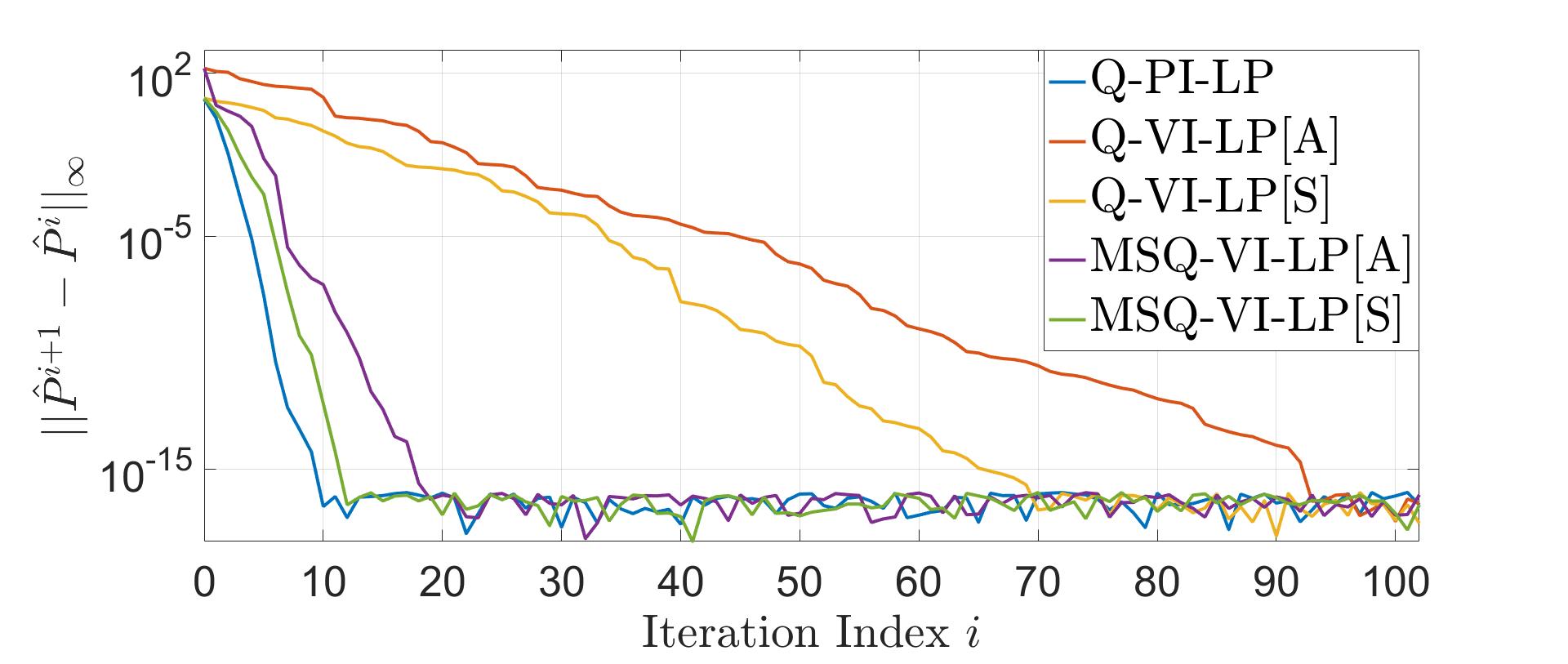}
\end{array}$
\end{centering}
\caption{Performance of LP algorithms in the unconstrained (top) and constrained-input (bottom) cases.}
\end{figure}
\\Figure $1$ shows the performance of all considered LP algorithms. The performance is evaluated based on the error between two successive Q-functions, i.e. by calculating the element-wise norm $||\hat{P}^{i+1}-\hat{P}^{i}||_{\infty}$. For the unconstrained and constrained-input cases, it is observed that:
\begin{itemize}
\item MSQ-VI-LP[A] converges in $15$ and $19$ iterations respectively,
\item MSQ-VI-LP[S] converges in $10$ and $13$ iterations respectively,
\item Q-PI-LP converges in $9$ and $12$ iterations respectively,
\item Q-VI-LP[A] converges in $84$ and $94$ iterations respectively, and finally
\item Q-VI-LP[S] converges in $55$ and $67$ iterations respectively. 
\end{itemize}
We note that in Q-PI-LP, we have to count an additional iteration which is not shown in the figure. This is because Q-PI-LP requires an iteration to compute $\hat{P}^{0}$, while for Q-VI-LP and MSQ-VI-LP $\hat{P}^{0}$ is given apriori during initialization. It is clear that the proposed MSQ-VI-LP algorithmic framework outperforms Q-VI-LP in terms of convergence speed. Furthermore, MSQ-VI-LP almost achieves the level of convergence speed of Q-PI-LP. An interesting observation is that the algorithmic variants with initial stabilizing policies benefit from an improved speed of convergence. In the unconstrained case, all LP algorithms converge to the Q-function defined by $\hat{P}^{\star}_{unc}$ given by
\begin{equation*}
\begingroup % keep the change local
\setlength\arraycolsep{0.05pt}
\begin{bmatrix} 1.4919 & -0.3188 & 1.6205 & -1.5628 & 1.1226 & 1.1514 & -0.4904\\ -0.3188 & 1.4633 & 1.0812 & 2.0894 & -0.6807 & -1.2905 & 0.5798 \\ 1.6205 & 1.0812 & 1.7562 & -1.4084 & 1.1803 & 0.8366 & -0.1929 \\ -1.5628 & 2.0894 & -1.4084 & 2.3127 & -1.1903 & -0.8835 & -0.41\\ 1.1226 & -0.6807 & 1.1803 & -1.1903 & 1.1439 & 0.9940 & 0.2128 \\ 1.1514 & -1.2905 & 0.8366 & -0.8835 & 0.9940 & 0.9997 & -0.3828 \\ -0.4904 & 0.5798 & -0.1929 & -0.41 & 0.2128 & -0.3828 & 1.2541 \end{bmatrix}
\endgroup
\end{equation*}      
while in the constrained-input case all LP algorithms converge to the Q-function defined by $\hat{P}^{\star}_{con}$ given by
\begin{equation*}
\begingroup % keep the change local
\setlength\arraycolsep{0.02pt}
\begin{bmatrix} 3.41 & 0.3425 & 3.1847 & -1.4843 & -0.2365  & 0.8246 & -0.353 \\ 0.3425 & 2.2134 & 0.3118 & -0.831 & -0.0006 & -0.3326 & 0.2298\\ 3.1847 & 0.3118 & 6.4281 & -4.2846 & 0.6503  & 1.0746 & -0.9133 \\ -1.4843 & -0.831 & -4.2846 & 5.4101 & -0.6801 & -0.8131 & 0.033\\ -0.2365 & -0.0006 &0.6503 &-0.6801 & 1.616 & -1.6175  & -0.388\\  0.8246 & -0.3326 & 1.0746 & -0.8131 & -1.6175 & 2.6316 & 0.0135 \\ -0.353 & 0.2298 & -0.9133 & 0.033 & -0.388 & 0.0135 & 1.9404\end{bmatrix}
\endgroup
\end{equation*} 
\\Figure $2$ shows the number of iterations which both variants of MSQ-VI-LP require to converge for the same threshold parameter ($\epsilon=10^{-15})$, as a function of the parameter $K$ used in the tuning law (26). Since for $K=0$ (i.e. $H_i=1$ for all $i$) the MSQ-VI-LP algorithm is converted to the Q-VI-LP algorithm, it is not surprising that MSQ-VI-LP[A] and MSQ-VI-LP[S] converge in the same number of iterations as Q-VI-LP[A] and Q-VI-LP[S] respectively in that scenario. As $K$ increases, the required number of iterations decreases, until the emperical bound $K=5$, after which we cannot decrease the required number of iterations further. This is an important result that will be theoretically investigated in future work.
\begin{figure}[htp]
\begin{centering}$
\hbox{\hspace{-0.25em}}\includegraphics[width=0.478\textwidth]{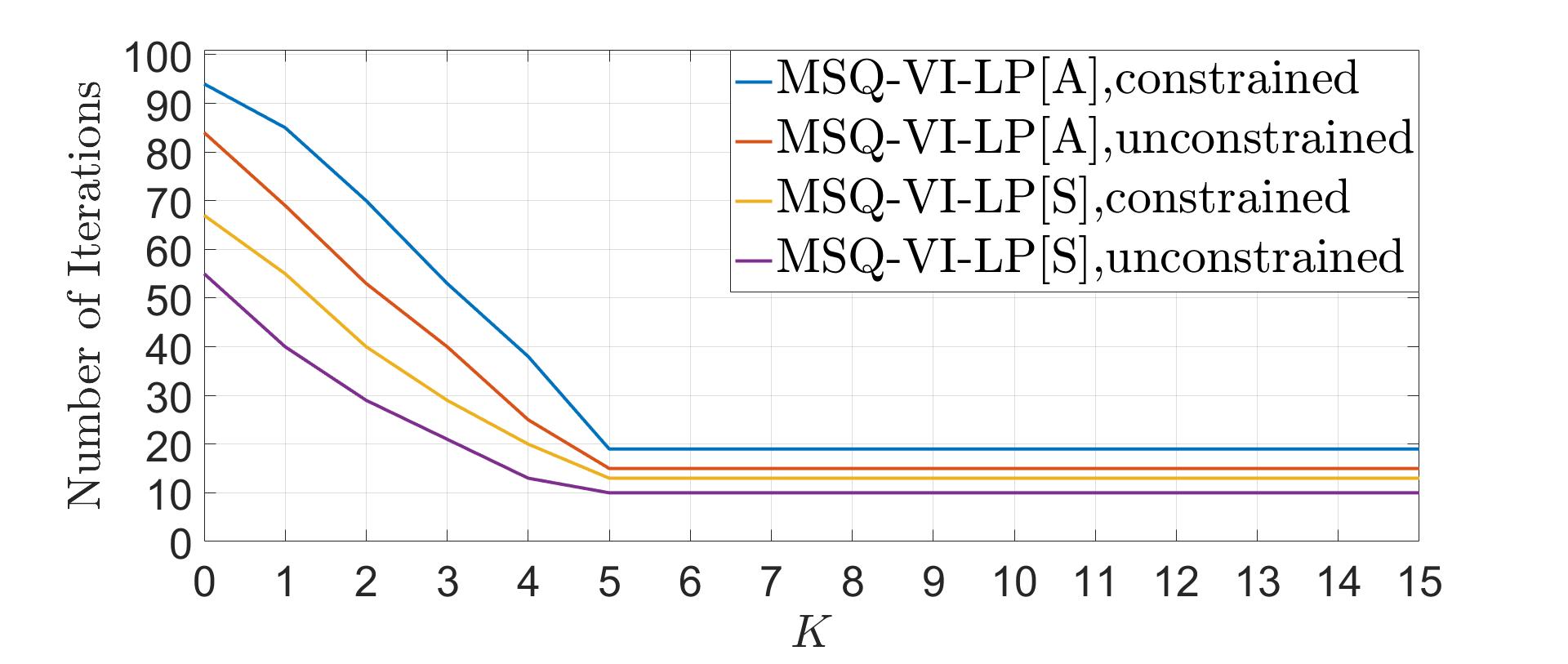}$
\end{centering}
\caption{Performance of MSQ-VI-LP variants as a function of $K$.}
\end{figure}
\\Figure $3$ shows the tracking control performance for the derived controllers.
\begin{figure}[htp]
\begin{centering}$
\begin{array}{ccc}
\hbox{\hspace{-0.0em}}\includegraphics[width=0.46\textwidth]{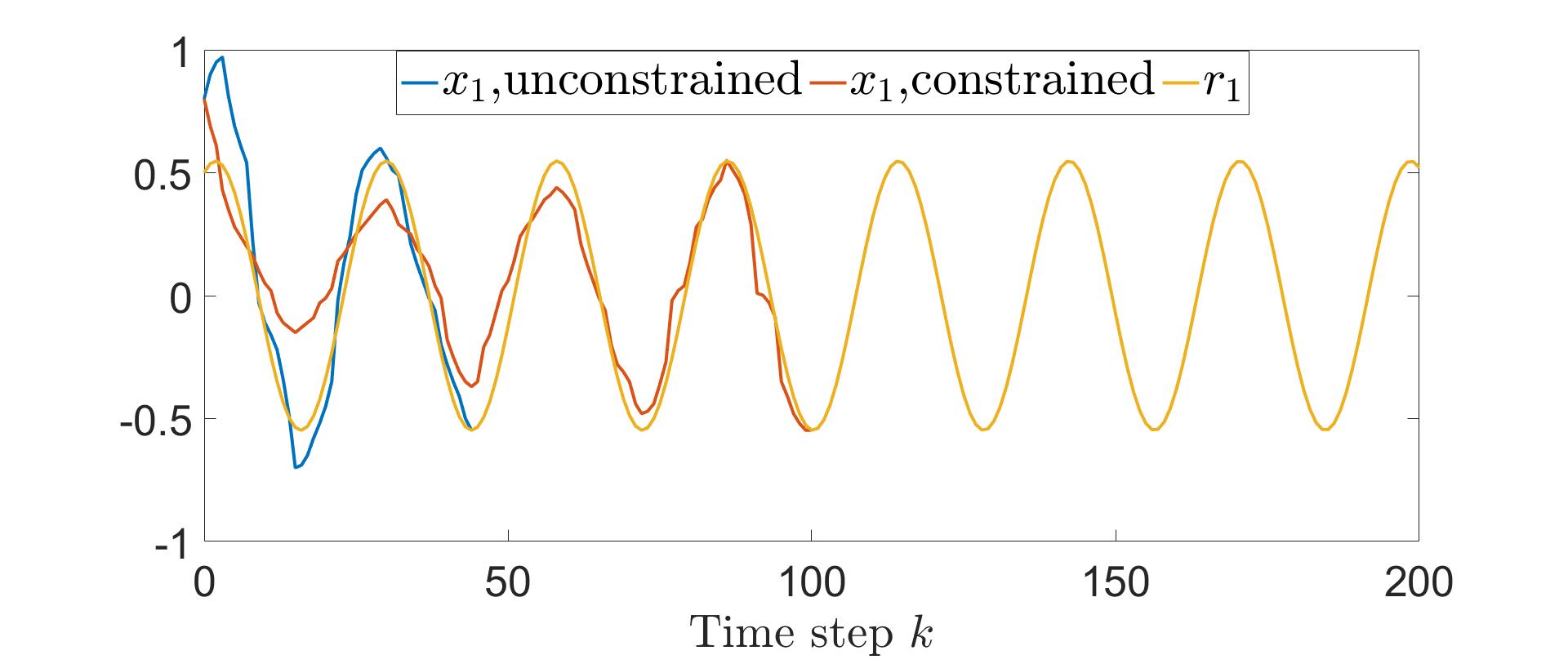} \\
\hbox{\hspace{-0.0em}}\includegraphics[width=0.46\textwidth]{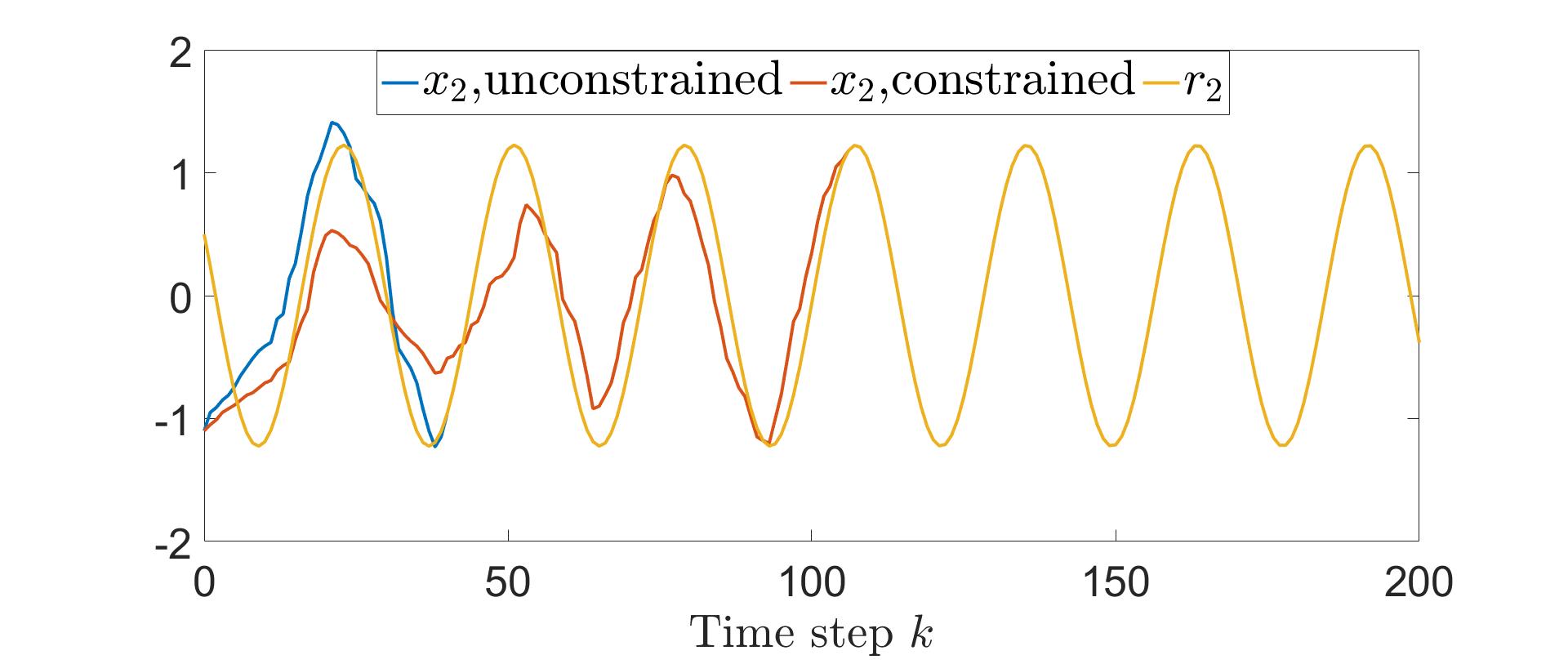}\\
\hbox{\hspace{-0.0em}}\includegraphics[width=0.46\textwidth]{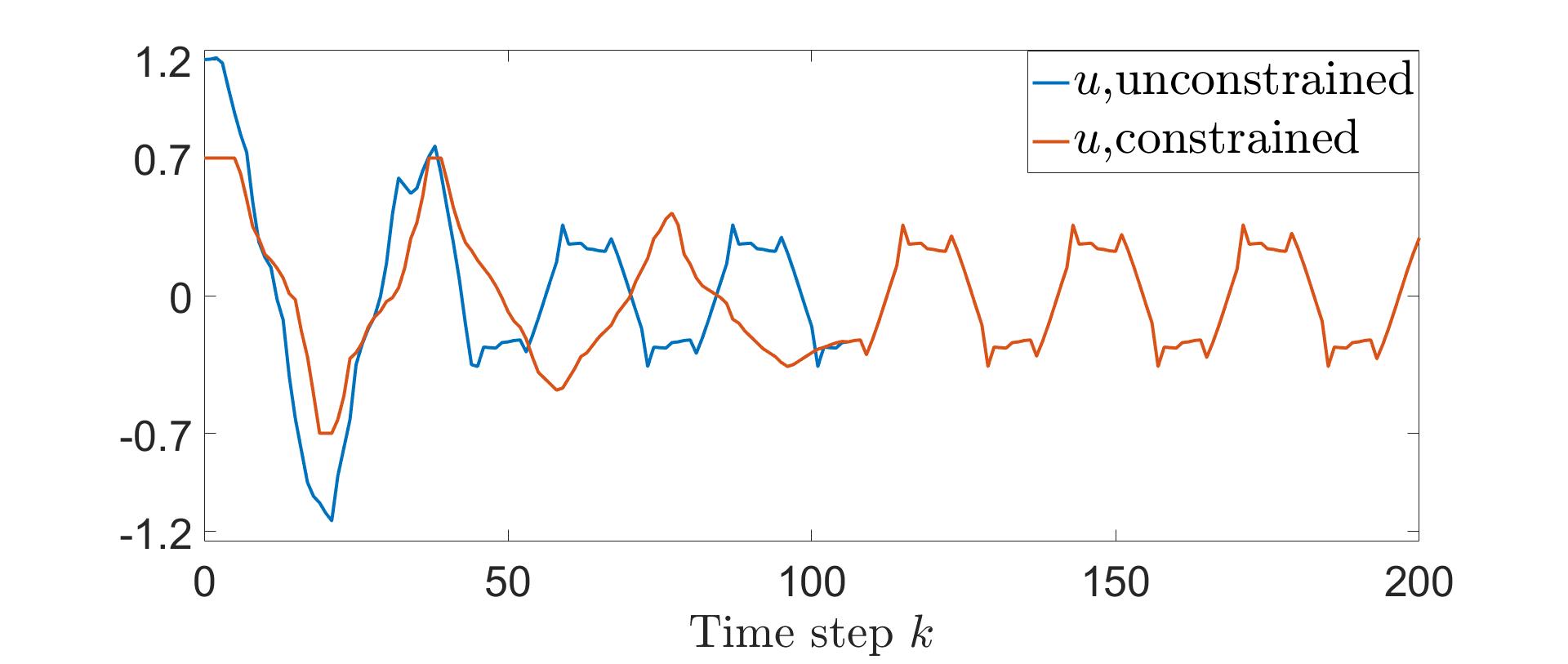}
\end{array}$
\end{centering}
\caption{Tracking control performance for both unconstrained and constrained-input cases, under the derived controllers.}
\end{figure}
It is observed that the unconstrained controller achieves optimal tracking control, but violates the considered input constraints. The derived constrained controller is capable of providing high-performance tracking behavior, while satisfying the imposed input constraints.
\section{Conclusions}
In this work, we derived a novel, high-performance multi-step VI algorithm based on Q-learning and linear programming, successfully extending the LP approach to ADP to the critical setting of optimal tracking control of general unknown discrete-time deterministic systems with general stage cost functions. We validated its performance in simulation, on both the unconstrained and constrained-input optimal tracking control of a nonlinear system. The success of the derived approach leads us to investigate possible extensions to other challenging domains, like the problem of optimal tracking control of unknown systems under adversarial disturbances using novel data-driven $H_{\infty}$ control algorithms.
 
\end{document}